%% file: 00-VoronoiMainLipics.tex
\documentclass[a4paper,UKenglish,cleveref, autoref, thm-restate]{lipics-v2019}

\usepackage[normalem]{ulem}
\nolinenumbers

\title{Competitive Location Problems: \\ Balanced Facility Location and the One-Round Manhattan Voronoi Game} 

\titlerunning{Competitive Location Problems}

\author{Thomas Byrne}{
Department of Management Science, University of Strathclyde, United Kingdom
}{tom.byrne@strath.ac.uk}{https://orcid.org/0000-0003-0548-4086}{}
\author{S\'{a}ndor P. Fekete}{Department of Computer Science, TU Braunschweig, Germany}{s.fekete@tu-bs.de}{https://orcid.org/0000-0002-9062-4241}{}
\author{Jörg Kalcsics}{
School of Mathematics, University of Edinburgh, United Kingdom
}{joerg.kalcsics@ed.ac.uk}{https://orcid.org/0000-0002-5013-3448}{}
\author{Linda Kleist}{Department of Computer Science, TU Braunschweig, Germany}{l.kleist@tu-bs.de}{https://orcid.org/0000-0002-3786-916X}{}

\authorrunning{T.\ Byrne, S.\,P.\ Fekete, J.\ Kalcsics, L.\ Kleist} 

\Copyright{Thomas M. Byrne, S\'{a}ndor P. Fekete, Jörg Kalcsics,  Linda Kleist}


\hideLIPIcs

\relatedversion{An extended abstract based on the content of this preprint appears at the International Conference and Workshops on Algorithms and Computation (WALCOM) 2021 \cite{walcomPaper}.}

\keywords{Facility location, competitive location, Manhattan distances, Voronoi game, geometric optimization}

\usepackage[dvipsnames]{xcolor}
\usepackage{nicefrac}
\usepackage{xspace}
\usepackage{cite}
\usepackage{subcaption}
\usepackage[normalem]{ulem}

\newcommand{\area}{\text{area}}
\newcommand{\bridge}{bridge\xspace}
\newcommand{\VD}{\mathcal V}
\newcommand{\B}{\mathcal{B}}
\newcommand{\balanced}{balanced\xspace}
\newcommand{\leaveout}[1]{{}}

\newtheorem{observation}[theorem]{Observation}
\crefname{observation}{Observation}{Observations}
\crefname{enumi}{}{}

\hypersetup{
	colorlinks = true,
	allcolors=blue,
	urlcolor=black
}

\graphicspath{{figures/}}

\newcommand{\new}[1]{#1}

\makeatletter
\let\orgdescriptionlabel\descriptionlabel
\renewcommand*{\descriptionlabel}[1]{%
	\let\orglabel\label
	\let\label\@gobble
	\phantomsection
	\edef\@currentlabel{#1}%
	\let\label\orglabel
	\orgdescriptionlabel{#1}%
}
\makeatother

\begin{document}

\maketitle

\begin{abstract}
We study competitive location problems in a continuous setting, in which facilities have to be placed in a rectangular domain $R$ of normalized dimensions of $1$ and $\rho\geq 1$, and distances are measured according to the Manhattan metric. 
We show that the family of \emph{balanced} facility configurations (in which the Voronoi cells of individual facilities are equalized with respect to a number of geometric properties) is considerably richer in this metric than for Euclidean distances. 
Our main result considers the \emph{One-Round Voronoi Game} with Manhattan distances, in which first  player White and then player Black each place $n$ points in $R$; each player scores the area for which one of its facilities is closer than the facilities of the opponent.
We give a tight characterization: White has a winning strategy if and only if $\rho\geq n$; for all other cases, we present a winning strategy for Black.
\end{abstract}

 \input{01-Intro}
 \input{02-Basics}
 \input{05-AreaSymmetric}

 \input{03-Black}
 \input{04-White}

 \input{06-Conclusion}

\bibliography{abbrv,VoronoiBib}

\end{document}

%% file: 01-Intro.tex
\section{Introduction}
\label{sec:intro}

Problems of optimal location are arguably among the most important in a wide
range of areas, such as economics, engineering, and biology, as well as in
mathematics and computer science. In recent years, they have gained a
tremendous amount of importance through clustering problems in artificial
intelligence.  In all scenarios, the task is to choose a set of positions from
a given domain, such that some optimality criteria with respect to the
resulting distances to a set of demand points are satisfied; in a geometric
setting, Euclidean or Manhattan distances are natural choices. Another
challenge of facility location problems is that they often happen in a
\emph{competitive} setting, in which two or more players contend for the best
locations.
A change to competitive, multi-player versions can have a serious impact on the
algorithmic difficulty of optimization problems: for example, the classic
Travelling Salesman Problem is NP-hard, while the competitive two-player
variant is even PSPACE-complete~\cite{fff+-tspc-04}.

In this paper, \new{we study the two-player} \emph{One-Round Voronoi Game} \new{in which} first player \emph{White} and then player \emph{Black} each place $n$ points \new{at once} in \new{a rectangle $R \subset \mathbb{R}^2$ of normalized dimensions with height $1$ and width $\rho\geq 1$}. Each player scores the area consisting of the points that are closer to one of their facilities than to any one of their opponent's; see \Cref{fig:examples} for an example. 
The goal \new{of} each player is to \new{win by obtaining} the higher score. \new{If both players obtain the same score,  the game ends in a \emph{draw}.}
\new{We note that we are not interested in the margin by which a player wins in the Voronoi game. This is a crucial difference to the classical leader-follower problem introduced by von Stackelberg, where each player seeks to maximize their score \cite{Sta52,SimCru73}.} \new{We assume that both players are fully aware of these rules, i.e., they know the parameter $n$ and each other's objective before the start of the game.} 

Exploiting the geometric nature of Voronoi cells, we completely resolve \new{this} classic problem of competitive location theory for the previously open case of Manhattan distances. 
\new{Despite the fact that they are} frequently studied in location theory and applications (e.g., see~\cite{kole,kus_nis, wl-lfrda-71}), \new{Manhattan distances} have received limited attention in a setting in which facilities compete for customers. While for Euclidean distances a bisector (the set of points that are of equal distance from two facilities) is the boundary of the open Voronoi cells and thus has area zero, Manhattan bisectors may have positive area, as shown in
\Cref{fig:bisector}.  \new{This results in the fact that}
both players may score strictly less than $\rho/2$, \new{with} the remaining area belonging to \emph{neutral zones}.
\new{A key role for analysing the game and characterizing winning strategies for the players falls to identifying \emph{balanced configurations}. The latter describe a set of points placed in $R$ such that all half cells of the respective Voronoi cells restricted to $R$ have equal area. A half cell is hereby the part of the Voronoi cell to the left or right of the vertical line through the point generating the cell or symmetrically above or below the horizontal line. In a balanced configuration, all Voronoi cells have equal area, yielding a \emph{fair} apportionment of the rectangle's area among the cells, and each cell generator minimizes the average distance to all points of the cell, meaning their position is \emph{locally optimal} within their cell.}

\medskip
\noindent 
Our main results are twofold. 
\begin{itemize}
	\item We show that for location problems with Manhattan distances in
	the plane, the properties of \emph{fairness} and \emph{local optimality} lead
	to a geometric condition called \emph{balancedness}. While the analogue concept
	for Euclidean distances in a rectangle implies grid
	configurations~\cite{fekete2005}, we demonstrate that there are \balanced
	configurations of much greater variety. 
	\item We give a full characterization of the One-Round Manhattan
	Voronoi Game \new{where each player places $n$ points} in a rectangle~$R$ with aspect ratio $\rho\geq 1$. 
	We show that White has a winning strategy if and only if $\rho\geq n$; for all other cases,
	Black has a winning strategy.
\end{itemize}

\begin{figure}[t]
	\centering
	\begin{subfigure}[t]{.27\textwidth}
		\centering
		\includegraphics[page=4]{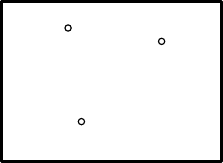}
		\caption{White places 3 points.}
	\end{subfigure}\hfil
	\begin{subfigure}[t]{.27\textwidth}
		\centering
		\includegraphics[page=5]{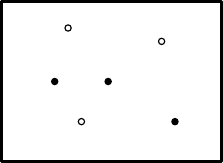}
		\caption{Black places 3 points.}
	\end{subfigure}\hfil
	\begin{subfigure}[t]{.27\textwidth}
		\centering
		\includegraphics[page=6]{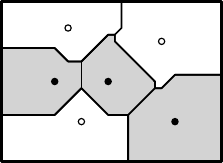}
		\caption{The dominated areas.}
	\end{subfigure}\hfil

	\caption{Example of a one-round Manhattan Voronoi game.}
	\label{fig:examples}
\end{figure}

\section{Related Work}
\label{sec:litreview}
Our paper relates to previous work in the field of competitive facility location in general and the Voronoi game in particular, both in a geometric setting.
\subsection{Competitive Facility Location}
Scientific work on facility location can be traced back to the turn of the 20th
century, with the groundbreaking works of Launhardt~\cite{Lau00} and
Weber~\cite{Web09}. Given two raw material suppliers and a single market, they
studied the problem of determining an optimal location for a new plant. See
Wesolowsky~\cite{Wes93} for more details on the history of facility location.
All of these early works dealt with continuous location problems, i.e., the new facilities can be located anywhere \new{in} the plane. Location problems in other
domains started to evolve in the 1960s, most notably on networks
(Hakimi~\cite{Hak64}) and discrete location problems (Manne~\cite{Man64}). While
the latter now dominate the literature, planar location problems have received
considerable and ongoing attention in the literature since those early works;
see the books of Drezner \cite{Dre95}, Drezner and Hamacher \cite{DreHam02}, and  Laporte et al.~\cite{LapNicSal19}.

The first discussion of \emph{competitive} facility location problems is due to
Hotelling~\cite{Hot29} in 1929. He considered the case of a bounded linear
market with uniform demand and two players who each locate one facility \new{from where a homogeneous commodity with no production cost is sold at a fixed price of their own choosing}.
\new{Assuming customers incur transportation costs that are linear in the distance travelled, it is presumed that customers patronize the facility from which they can purchase one unit of the commodity at the lowest price. For this set-up, Hotelling claimed} that there is
an equilibrium solution for which both facilities are located right next to each
other in the centre of the line segment, with each player capturing half of the demand. \new{This was later disputed by d’Aspremont et al.~\cite{dAsGabThi79}, who proved that no equilibrium pricing solution exists when players co-locate. However, 
these issues stem from the inclusion of the pricing problem within the underlying modelling assumptions; the principle of minimum differentiation, commonly referred to as Hotelling's law, holds within the pure location model. Those interested in a further discussion are encouraged to read Aydinonat and K\"{o}ksal~\cite{AydKok19} and references therein.}
Since Hotelling's seminal work, numerous other competitive location models have appeared in the literature; 
we refer the interested reader to Dasci~\cite{Das11} \new{and} Eiselt et al.~\cite{Eiselt2019}.

Drezner~\cite{Dre82} considers a given finite set of
customer points in the plane, each with their own demand. 
First the \emph{leader}
(player 1) and then the \emph{follower} (player 2) place a fixed number of facilities.
The market share of a facility is then given by the total demand of all
customers who are closer to this facility than to any of the other facilities\new{; the leader's facility is considered to be closer in case of equal distance.}
Given the facility locations $W$ of the leader, the follower wants to place his
$r$ facilities so as to maximize his total market share. The leader, in turn,
wants to place her $p$ facilities such that her market share is maximal after
the follower places his facilities optimally. \new{Hakimi~\cite{Hak83}} called the former the
\emph{follower's} or $(r|W)$-medianoid problem, and the latter the
\emph{leader's} or $(r|p)$-centroid problem. \new{For Euclidean distances,} he derived exact polynomial-time
algorithms for the $(r|1)$-centroid and the $(1|W)$-medianoid problem. Not much
is known about the general version of these problems.
Bhadury et al.~\cite{BhaEisJar03} present a heuristic for the $(r|p)$-centroid
problem, in which the two players alternate in solving a 
medianoid problem. 
\new{For the $(r|W)$-medianoid problem}, the authors propose two heuristics: one based on incrementally solving the $(1|W)$-medianoid problem and the other based on placing the follower's facilities right next to the leader's sites. 
For the \new{Manhattan} metric, Infante-Macias and Mu\~{n}oz-P\'{e}rez \cite{InfMun95}
derive an exact enumeration algorithm to solve the $(r|W)$-medianoid problem --
albeit in time exponential in $r$.

\new{Less is known for geometric location problems for which demand is not
discrete, but uniformly distributed across the plane. In this setting,}
Averbakh et al.~\cite{AveBerKalKra15} derive an exact polynomial-time algorithm
for the  $(1|W)$-medianoid problem \new{with Manhattan metric} (as well as for several non-competitive problems), finding an optimal location for
an additional facility in a convex region 
with $n$ existing facilities.
For Euclidean \new{metric}, a convex compact market area, and two players, each
placing one facility, Aoyagi and Okabe~\cite{AoyOka93} prove that an
equilibrium configuration exists if and only if the market area is point-wise
symmetric with respect to some point in the area; this point is then the
optimal location for both facilities. This is the two-dimensional analogon to
Hotelling's observation for a linear market. For more than two facilities,
however, the equilibrium configurations are markedly different. For an
unbounded plane and an infinite number of facilities (or competing players),
Okabe and Aoyagi~\cite{OkaAoy91} show that the equilibrium configuration forms
a regular hexagonal pattern. For a square market and a finite number of firms,
Okabe and Suzuki~\cite{OkaSuz87} show that the equilibrium state exhibits a honeycomb pattern. 

\new{The problem of finding a \emph{fair} apportionment of the rectangle's area, i.e., all Voronoi cells have equal area, has been addressed in Baron~et~al.~\cite{BarBerKraWan07} for  Euclidean distances. They consider a unit square with uniformly distributed demand and assume that each demand point is served by the closest facility, provided this facility is within a predefined radius of the demand point. The goal is then to locate $n$ facilities in the square such that all facilities face an equal demand load, the latter being computed as the area served exclusively by the facility. The authors propose an iterative procedure utilizing Voronoi diagrams to approximately solve the problem. Suzuki and Drezner~\cite{DreSuz09} present an improved gradient search approximation algorithm for this problem as well as for two new problems. We point out that configurations obtained by those problems are \emph{fair}, but not necessarily \emph{locally optimal} (and therefore not necessarily \emph{balanced}).}

For urban location problems, the Manhattan metric provides a much better \new{geometric} approximation of the actual travel distances than the Euclidean metric, even if urban road networks are not all grid-shaped. Moreover, many applications arise from multi-dimensional data sets with heterogeneous dimensions, where the Manhattan metric is a compelling choice. While rectilinear problems have been frequently studied in location
theory and applications (e.g., see~\cite{kole,kus_nis, wl-lfrda-71}), they have
received limited attention in a setting in which facilities compete for
customers. For non-competitive
problems with Manhattan distances, Fekete et al.~\cite{fmb-cfwp-05} provide several
algorithmic results, including an NP-hardness proof for the $k$-median problem
of minimizing the average distance. Along similar lines, Bender et al.~\cite{bbd-wosc-04}
describe the shape of a region with a desired area that minimizes
the normalized Manhattan distance; this characterization is based on a 
differential equation for which no closed-form solution is known, highlighting
the surprising depth of location problems with Manhattan distances.

\subsection{The Voronoi Game}
An important scenario for competitive facility location is the \emph{Voronoi game}, first introduced by Ahn et al.~\cite{Ahn2004}, in which two players, \emph{White} and \emph{Black}, take turns placing one facility at a time \new{in a given playing arena}. In the end, each player scores the area consisting of the points that are strictly closer to one of their facilities than to any one of the opponent's, i.e., the total area of their open Voronoi cells. 
The goal for each player is to obtain the higher score. 
As Teramoto et al.~\cite{teramoto2006voronoi} showed, the
problem is PSPACE-complete, even in a discrete graph setting. 
There is a considerable amount of other work on variants of the Voronoi game.
Bandyapadhyay et al.~\cite{bandyapadhyay2015voronoi} consider the one-round game in trees, providing a polynomial-time algorithm for the second player. As
Fekete and Meijer~\cite{fekete2005} have shown, the problem is NP-hard for
polygons with holes, corresponding to a planar graph with cycles. For a
spectrum of other variants and results,
see~\cite{banik2013one,durr2007nash,gerbner2013advantage,kiyomi2011voronoi}.

Special attention has been paid to the \emph{One-Round Voronoi Game}, in which
each player places their $n$ facilities at once, starting with White; see \Cref{fig:examples} for an example with Manhattan distances \new{and a rectangular arena}. 
\new{During the game, co-location of points is forbidden (or does not occur in optimal play when
breaking ties of equal distances in favour of white).
}

This Voronoi game  -- which can also be considered a special case of the $(n|n)$-centroid problem -- was first studied in Cheong et al.~\cite{cheong2004one}.
They showed that for Euclidean distances,
the first player, White, can always win \new{when the playing arena consists of a segment}, while the second player, Black, has a winning strategy if the \new{arena} is a square and $n$ is sufficiently large.
Fekete and Meijer~\cite{fekete2005} refined this by showing that in a rectangle
of dimensions $1\times \rho$ with $\rho\geq 1$, Black has a winning strategy
for $n\geq 3$ and $\rho<n/\sqrt{2}$, and for $n=2$ and $\rho<2/\sqrt{3}$; White
wins in all other cases.

Consideration of Manhattan distances leads to a number of important differences.
While for Euclidean distances a bisector (the set of points that are of equal
distance from two facilities) is the boundary of the open Voronoi cells, so its
area is zero, Manhattan bisectors may have positive area, as shown in
\Cref{fig:bisector}. Thus, owing to the different nature of the Manhattan
metric, both players may dominate strictly less than $\rho/2$, the remaining
area belonging to \emph{neutral zones}.

%% file: 02-Basics.tex
\section{Preliminaries}
\label{sec:voronoi}

Let $P$ denote a finite set of points in a rectangle $R$.
For two points $p_1=(x_1,y_1)$ and $p_2=(x_2,y_2)$, we define $\Delta_x(p_1,p_2):=|x_1-x_2|$ and $\Delta_y(p_1,p_2):=|y_1-y_2|$. Then their Manhattan distance  is given by 
$d_M(p_1,p_2):=\Delta_x(p_1,p_2)+\Delta_y(p_1,p_2)$.

Defining $D(p_1,p_2):=\{p\in R\mid d_M(p,p_1)< d_M(p,p_2)\}$ as a set of points that are closer to $p_1$ than to $p_2$, the \emph{Voronoi cell} of $p$ in $P$ is 
\begin{equation*}
V^P(p):=\bigcap_{q\in P\setminus \{p\}} D(p,q).
\end{equation*}
The \emph{Manhattan Voronoi diagram} $\VD(P)$ is the complement of the union of all Voronoi cells of~$P$. 
In contrast to the Euclidean case, for which the Voronoi diagram has measure zero and every Voronoi cell is convex:
\begin{itemize}
	\item the Manhattan Voronoi diagram may  contain \emph{neutral zones} of positive measure, and
	\item Manhattan Voronoi cells need not be convex, but they are star-shaped.
\end{itemize}

Both of these properties can easily be observed when analysing the bisectors. 
The \emph{bisector} of $p_1$ and $p_2$ is the set of all points that are of equal distance from $p_1$ and $p_2$, i.e.,
\begin{equation*}
\B(p_1,p_2):=\{q\in R \mid d_M(q,p_1)= d_M(q,p_2)\}.
\end{equation*}
There are three types of bisectors, as shown in \cref{fig:bisector}.
Typically, a bisector consists of  three one-dimensional parts, namely two (vertical or horizontal) segments that are connected by a segment of slope $\pm 1$; see \cref{fig:bisector1}.
 If $\Delta_x(p_1,p_2)=0$ or $\Delta_y(p_1,p_2)=0$, then the diagonal segment  shrinks to a point and the bisector consists of a (vertical or horizontal) segment; see \cref{fig:bisector2}.
However, when $ \Delta_x(p_1,p_2)= \Delta_y(p_1,p_2)$, then the bisector $\B(p_1,p_2)$ contains two regions; see \cref{fig:bisector3}. We call a bisector of this type \emph{degenerate}.
Further, a non-degenerate bisector is \emph{vertical} (\emph{horizontal}) if it contains vertical (horizontal) segments.
\begin{figure}[htb]
	\centering
	\begin{subfigure}[t]{.3\textwidth}
		\centering
		\includegraphics[page=3]{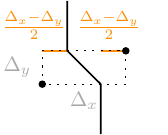}
		\caption{General vertical bisector.}
		\label{fig:bisector1}
	\end{subfigure}\hfill
	\begin{subfigure}[t]{.3\textwidth}
		\centering
		\includegraphics[page=4]{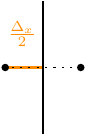}
		\caption{Case: $\Delta_y(p_1,p_2)=0$.}
		\label{fig:bisector2}
	\end{subfigure}\hfill
	\begin{subfigure}[t]{.3\textwidth}
		\centering
		\includegraphics[page=5]{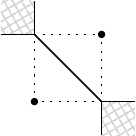}
		\caption{Degenerate bisector.}
		\label{fig:bisector3}
	\end{subfigure}
	\caption{Illustration of the three types of bisector.}
	\label{fig:bisector}
\end{figure}

For $p=(x_p,y_p)\in P$, both the vertical line $\ell_v(p)$
and the horizontal line $\ell_h(p)$
 through~$p$ split the Voronoi cell~$V^P(p)$ into two pieces, 
which we call \emph{half cells}.
We denote the set of all half cells of $P$ obtained by vertical lines
by $\mathcal H^\vert$ and those obtained by horizontal lines by $ \mathcal
H^-$. Furthermore, we define $\mathcal H:=\mathcal H^\vert\cup
\mathcal H^-$ as the set of all half cells of $P$. 
Applying both $\ell_v(p)$ and $\ell_h(p)$ to~$p$ yields a subdivision into four quadrants, which we denote by $Q_i(p)$, $i\in\{1,\ldots,4\}$; see \Cref{fig:quadrant}.  Moreover, $C_i(p):=V^P(p)\cap Q_i(p)$ is called the $i$th \emph{quarter cell} of~$p$.
We also consider the eight regions of every
$p\in P$ obtained by cutting $R$ along the lines $\ell_v(p)$,  $\ell_h(p)$, and the two
diagonal lines of slope $\pm 1$ through $p$. We refer to each such (open) region as an
\emph{octant} of $p$ denoted by $O_i(p)$ for $i\in\{1,\ldots,8\}$ 
(see \Cref{fig:octant}); a closed octant is denoted by $\overline
O_i(p)$. The area of a subset $S$ of $R$ is denoted by $\area(S)$.

\begin{figure}[htb]
	\centering
		\begin{subfigure}[t]{.4\textwidth}
		\centering
		\includegraphics[page=11]{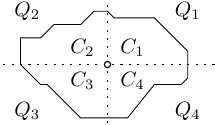}
		\caption{The quadrants and quarter cells.}
		\label{fig:quadrant}
	\end{subfigure}\hfil
	\begin{subfigure}[t]{.2\textwidth}
		\centering
		\includegraphics[page=2]{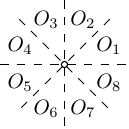}
		\caption{The octants.}
		\label{fig:octant}
	\end{subfigure}\hfil
	\begin{subfigure}[t]{.3\textwidth}
		\centering
		\includegraphics[page=34]{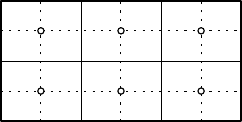}
		\caption{A $2\times 3$ grid.}
		\label{fig:AreaBal}
	\end{subfigure}
	\caption{Illustration of crucial definitions.}
	\label{fig:misc}
\end{figure}

For a point $p\in P$, we call the four horizontal and vertical rays rooted at $p$, contained within $V^P(p)$, the four \emph{arms} of $V^P(p)$ (or of $p$).
Two arms are \emph{neighbouring} if they appear consecutively in the cyclic order; otherwise they are \emph{opposite}. Moreover, we say an arm is a \emph{boundary arm}  if its end point touches the boundary of $R$; otherwise it is \emph{inner}.
For later reference, we note the following. 
\begin{observation}\label{obs:prop}
	The following properties hold:
	\begin{enumerate}[(i)]
		\item \label{itemA} If the bisector $\B(p,q)$ is non-degenerate and vertical (horizontal), then it does not intersect both the left and right (top and bottom) half cells of~$p$.
		\item \label{itemB} For every $i$ and every $q_1,q_2\in O_i(p)$, the bisectors $\B(p,q_1)$ and $\B(p,q_2)$ have the same type (vertical/horizontal). 
		\item \label{itemC} A Voronoi cell is contained in the axis-aligned rectangle spanned by its arms.
	\end{enumerate}
\end{observation}

\begin{proof}
	Properties~(\cref{itemA}) and (\cref{itemB}) follow immediately from the shape of the bisectors. 
 Property~(\ref{itemB}) implies that a Manhattan Voronoi cell consists of four ($x$- and $y$-) monotone paths connecting the tips of its arms. Consequently, each Voronoi cell is contained in the axis-aligned rectangle spanned by its arms and, thus, property~(\ref{itemC}) holds.
\end{proof}

%% file: 05-AreaSymmetric.tex
\section{Balanced Point Sets}\label{sec:areaBalanced}

In this section, we discuss properties of balanced point sets in the Manhattan metric (\cref{sec:balanced_char}) and present families of non-grid balanced point sets showing that they are considerably richer in the Manhattan metric than in the  Euclidean case (\cref{sec:balanced_nongrid,sec:balanced_family}).

\subsection{Properties and Characterization}\label{sec:balanced_char}
In a competitive setting for facility location, it is a natural \emph{fairness property}
to allocate the same amount of influence to each facility.
A second \emph{local optimality property} arises from choosing an efficient location
for a facility within its individual Voronoi cell, \new{i.e, a location that minimizes the average distance to all points}.
Combining both properties, we say 
a point set~$P$ in a rectangle~$R$ is \emph{\balanced} if  the
following two conditions are satisfied:
\begin{description}
\item[Fairness:] 
for all $p_1,p_2\in P$,  $V^P(p_1)$ and $V^P(p_2)$  have the same area.
\item [Local optimality:] for all $p\in
P$, $p$ minimizes the average distance to the points in $V^P(p)$.
\end{description}

For Manhattan distances, there is a  simple geometric characterization for the
local optimality depending on the area of the half and quarter cells; see \Cref{fig:quadrant}.

\begin{restatable}{lemma}{median}
\label{obs:quadrant}
A point $p$ minimizes the average Manhattan distance to the points in $V^P(p)$ if and only if either \new{one} of the following properties holds:
\begin{enumerate}[(i)]
	\item $p$ is a \emph{Manhattan median} of $V^P(p)$: all four half cells of $V^P(p)$ have the same area. \label{item:halfcells}
	\item $p$ satisfies the \emph{quarter-cell property}: diagonally opposite quarter cells of $V^P(p)$ have the same area.
	 \label{item:quadrant}
\end{enumerate}
\end{restatable}

\begin{proof}
	For an illustration we refer to \Cref{fig:quadrant}.
	Let $a_i$ denote the area of the quarter cell $C_i(p)$. First we
	consider a point $p=(x_p,y_p)$ that minimizes the average Manhattan distance to
	all points in $V^P(p)$. Suppose the area of the top half cell exceeds the area
	of the bottom half cell, i.e., $a_1+a_2>a_3+a_4$. Then, replacing $p$ by
	$p'=(x_p,y_p+\varepsilon)$ for an appropriately small $\varepsilon>0$ reduces
	the average $y$-distance and leaves the average $x$-distance unchanged. This contradicts the optimality of $p$. Similarly,
	we can exclude $a_1+a_2<a_3+a_4$, so $a_1+a_2=a_3+a_4$, making $p$ a $y$-median of $V^P(p)$.
	Analogously, we conclude that $a_1+a_4=a_2+a_3$, making $p$ an $x$-median of $V^P(p)$. This shows that all half cells of $V^P(p)$ have the same area, \new{i.e., property (\ref{item:halfcells}) holds}.
	Moreover, note that the half-cell condition uniquely defines both $x_p$ and $y_p$, so \new{(\ref{item:halfcells})} is both necessary and sufficient.
	
	We now show that (\ref{item:halfcells}) is equivalent to (\ref{item:quadrant}).
	By adding the equations 
	\begin{align*}
	a_1+a_2&=a_3+a_4\\
	a_1+a_4&=a_2+a_3
	\end{align*}
	it follows that 
	$2a_1+a_2+a_4=2a_3+a_2+a_4\iff a_1=a_3$.
	By subtracting the equations, we get
	$a_2-a_4=a_4-a_2\iff a_2=a_4.$ Hence, the quarter-cell property is fulfilled. 
	
	Conversely,  $a_1=a_3$ and $a_2=a_4$ imply $a_1+a_2=a_3+a_4=a_1+a_4=a_2+a_3$.
\end{proof}

\Cref{obs:quadrant} immediately implies the following characterization.

\begin{corollary}
\label{th:balanced}
A point set $P$ in a rectangle $R$ is \balanced if and only if all half cells of $P$ have the same area. 
\end{corollary}

\subsection{Atomic Non-Grid Configurations}\label{sec:balanced_nongrid}
A simple family of \balanced sets arise from regular, $a\times b$ grids;  see \Cref{fig:AreaBal}.
In stark contrast to the Euclidean case, there exist a large variety of
other \balanced sets: \Cref{fig:equalArea} depicts balanced point sets for which 
\emph{no} cell is a rectangle. 

\textcolor{white}{this is a test}\vspace{-12pt}

\begin{figure}[htb]
	\centering
	\includegraphics[page=38]{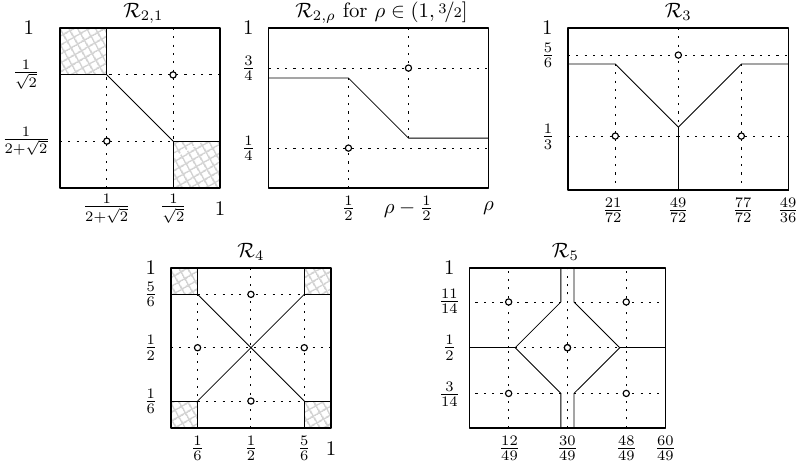}
	\caption{Non-grid examples of balanced point sets of cardinality $2$, $3$, $4$, and $5$.}
	\label{fig:equalArea}
\end{figure}

\begin{restatable}{lemma}{unique}\label{lem:unique}
The configurations $\mathcal R_{2,\rho},\mathcal R_{3},\mathcal R_{4},\mathcal R_{5}$, depicted in \Cref{fig:equalArea}, are \balanced. 
Moreover, $\mathcal R_{2,\rho},\rho \in [1,\nicefrac{3}{2}]$, and $\mathcal R_{3}$
are the only \balanced non-grid point sets  with two and three points, respectively.
\end{restatable}

Simple calculations show that the configurations are \balanced. In order to prove the uniqueness, we make use of \Cref{obs:quadrant}. While the analysis for $n=2$ can be easily conducted manually, for $n=3$ the relative point positions lead to about 20 cases of structurally different Voronoi diagrams, which were checked using MATLAB\textsuperscript{\textregistered}. 
In the following, we provide proof details.
\Cref{clm:uniqueness2} in \Cref{subsub:2} establishes uniqueness for $n=2$; \Cref{clm:uniqueness3} 
in \Cref{subsub:3} shows uniqueness for $n=3$.

\subsubsection{Balanced Sets with Two Points}
\label{subsub:2}
In the following, we establish uniqueness for $n=2$.
\begin{claim}\label{clm:uniqueness2}
	For every $(1\times \rho)$ rectangle with $1\leq \rho\leq \nicefrac{3}{2}$, there exists (up to reflection) a unique point set $P$ with  $|P|=2$ such that $P$ is not a grid and fulfils \ref{item:P1}. The resulting configuration is $\mathcal R_{2,\rho}$ as illustrated in \Cref{fig:equalArea}.
	
	Moreover, if $\rho> \nicefrac{3}{2}$, there exists no such point set.
\end{claim}

\begin{claimproof}
Let $P$ be a point set in an $(h\times w)$ rectangle $R$ consisting of two points $p_1=(x_1,y_1)$ and $p_2=(x_2,y_2)$ which is different from a grid and fulfils \ref{item:P1}.
Without loss of generality, we assume that $p_2$ lies in the top right quadrant of $p_1$ as in \cref{fig:equalAreaN2}.
We distinguish two cases depending on whether or not the vertical and horizontal distances between $p_1$ and $p_2$ are equal.

 \begin{figure}[htb]
	\centering
	\begin{subfigure}[t]{.4\textwidth}
		\centering
		\includegraphics[page=33]{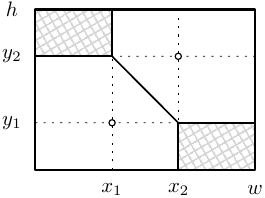}
		\caption{Case: $|x_1-x_2|=|y_1-y_2|$}
		\label{fig:equalAreaN2A}
	\end{subfigure}
	\hfil
	\begin{subfigure}[t]{.4\textwidth}
		\centering
		\includegraphics[page=21]{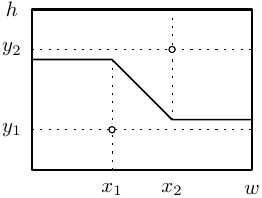}
		\caption{Case: $|x_1-x_2|<|y_1-y_2|$}
		\label{fig:equalAreaN2B}
	\end{subfigure}
	\caption{Illustration of the proof of \Cref{clm:uniqueness2}.}
	\label{fig:equalAreaN2}
\end{figure}

Firstly, we consider the case that $\Delta_x(p_1,p_2)=\Delta_y(p_1,p_2)=:d$. For an illustration, see \cref{fig:equalAreaN2A}.
By \ref{item:P1} and \Cref{obs:quadrant}, diagonally opposite quarter cells have the same area. Consequently, the second and fourth rectangular quarter cells of $p_1$ and $p_2$
imply that $dx_1=dy_1 \iff x_1=y_1$,  and 
 $d(w-x_2)=d(h-y_2) \iff w-x_2=h-y_2$. 
 Hence, $w=x_1 + d + (w-x_2) = y_1 + d + (h-y_2) = h$ so it follows that $\rho = 1$.
 This yields $\mathcal R_{2,1}$ as shown in \Cref{fig:equalArea}.

Secondly, we consider the case that $\Delta_x(p_1,p_2)\neq\Delta_y(p_1,p_2)$.
Without loss of generality, we assume that $\Delta_x(p_1,p_2)<\Delta_y(p_1,p_2)$ as illustrated in \Cref{fig:equalAreaN2B}. Thus, $\B(p_1,p_2)$ is horizontal. By \ref{item:P1}, the bottom half cell of $p_1$ and the top half cell of $p_2$ have an area of $\nicefrac{1}{4}\, wh$ each. Because their width is $w$, it follows that $y_1=\nicefrac{1}{4}h$ and $y_2=\nicefrac{3}{4}h$.
By symmetry of the bisector, it follows that the height of the left half cell of $p_1$ equals the height of the right half cell of $p_2$. Because the areas of these half cells are equal, their respective widths must  also agree,  i.e., $x_2=w-x_1$.
Moreover, the left half cell of $p_1$ has an area of 
\begin{align*}
x_1(\nicefrac{h}{4}+\nicefrac{1}{2}(\nicefrac{h}{2}+(x_2-x_1))=\nicefrac{x_1}{2}(h+x_2-x_1)&=\nicefrac{x_1}{2}(h+w-2x_1)\stackrel{!}{=}\nicefrac{wh}{4}\\
&\iff x_1\in \{\nicefrac{h}{2}, \nicefrac{w}{2}\} \, .
\end{align*}

If $x_1=\nicefrac{w}{2}$, $P$ is a grid. For $x_1=\nicefrac{h}{2}$, we obtain the configuration $\mathcal R_{2,\rho} $ depicted in \Cref{fig:equalArea}. Note that it is necessary that $\nicefrac{h}{2}< w-\nicefrac{h}{2}\iff h< w$ and $(y_2-y_1)\geq (x_2-x_1)\iff w\leq \nicefrac{3}{2}h$. This completes the proof of the claim.
\end{claimproof}

\subsubsection{Balanced Sets with Three Points}
\label{subsub:3}
In the following, we establish uniqueness for $n=3$.
\begin{claim}\label{clm:uniqueness3}
	Let $P$ be a point set in a rectangle $R$ such that $|P|=3$, $P$ is not a grid, and $P$ satisfies \ref{item:P1}.  Then $\rho(R)=\nicefrac{49}{36}$ and $(R,P)$ is the configuration $\mathcal R_3$.
\end{claim}

\begin{claimproof}

We denote the height and width of $R$ by $h$ and $w$ respectively, and distinguish four cases depending on the number of degenerate bisectors of the points in $P$. We start with the case of no degenerate bisectors. 

\subparagraph{No degenerate bisectors.} Firstly, let us assume that $P$ contains no degenerate bisector. Therefore, every corner of $R$ is contained in one of the three cells and the cell of one point $p_1=(x_1,y_1)$ contains two corners of $R$; without loss of generality, we assume that $p_1$
contains the two top corners. Then, the top half cell of $p_1$ has width~$w$ and an area of $\nicefrac{1}{6}\, wh$ by property \ref{item:P1}.
Consequently, $y_1=\nicefrac{5}{6}\, h$. Moreover, the other two points lie in $O_6(p_1)\cup O_7(p_1)$; otherwise the cell of $p_1$ would not contain both top corners. 

\newcommand{\s}{.88}
\begin{figure}[tp]
	\centering
	\begin{subfigure}{.35\textwidth}
		\centering
		\includegraphics[page=28,scale=\s]{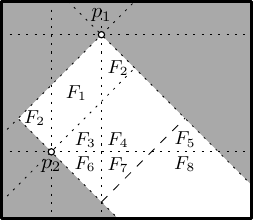}
		\caption{Configurations where $x_1 > x_2$.}
		\label{fig:nongrid3A}
	\end{subfigure}\hfil
	\begin{subfigure}{.35\textwidth}
		\centering
		\includegraphics[page=29,scale=\s]{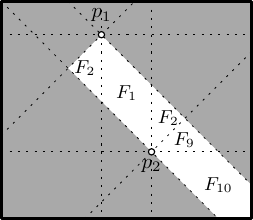} 
		\caption{Configurations where $x_1 \leq x_2$.}
		\label{fig:nongrid3B}
	\end{subfigure}
	\caption{Illustration of the cases in the proof of \Cref{clm:uniqueness3}.}
	\label{fig:nongrid3}
\end{figure}

\newcommand{\w}{.27\textwidth}
\begin{figure}[htbp]
	\centering
	\begin{subfigure}{\w}
		\centering
		\includegraphics[page=40,scale=\s]{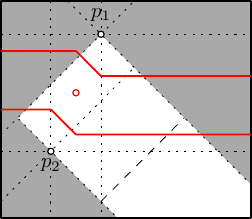}
		\caption{$p_3$ in $F_1$}
		\label{fig:n=3_F1}
	\end{subfigure}\hfil
	\begin{subfigure}{\w}
		\centering
		\includegraphics[page=41,scale=\s]{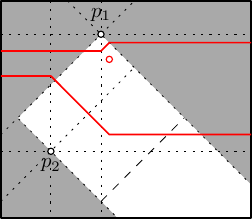}
		\caption{$p_3$ in $F_2$}
		\label{fig:n=3_F2}
	\end{subfigure}\hfil
	\begin{subfigure}{\w}
		\centering
		\includegraphics[page=42,scale=\s]{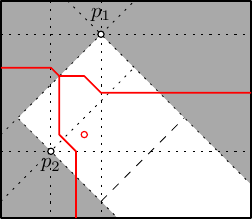}
		\caption{$p_3$ in
			$F_3$}
		\label{fig:n=3CellIII}
	\end{subfigure}\hfil
	\begin{subfigure}{\w}
		\centering
		\includegraphics[page=43,scale=\s]{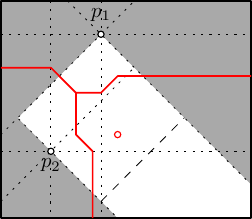}
		\caption{$p_3$ in 
			$F_4$}
		\label{fig:n=3CellIVa}
	\end{subfigure}\hfil
	\begin{subfigure}{\w}
		\centering
		\includegraphics[page=44,scale=\s]{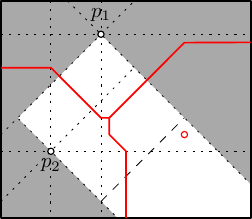}
		\caption{$p_3$ in 
			$F_5$}
		\label{fig:n=3CellIVb}
	\end{subfigure}\hfil
	\begin{subfigure}{\w}
		\centering
		\includegraphics[page=45,scale=\s]{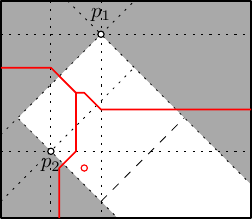}
		\caption{$p_3$ in 
			$F_6$}
		\label{fig:n=3CellI}
	\end{subfigure}\hfil
	\begin{subfigure}{\w}
		\centering
		\includegraphics[page=46,scale=\s]{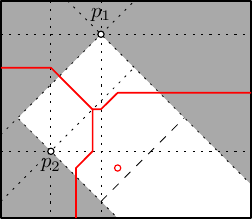}
		\caption{$p_3$ in 
			$F_7$}
		\label{fig:n=3CellIIa}
	\end{subfigure}\hfil
	\begin{subfigure}{\w}
		\centering
	\includegraphics[page=47,scale=\s]{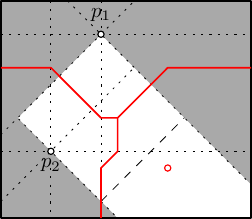}
		\caption{$p_3$ in 
			$F_8$}
		\label{fig:n=3CellIIb}
	\end{subfigure}
	
	\begin{subfigure}{.3\textwidth}
		\centering
		\includegraphics[page=49,scale=\s]{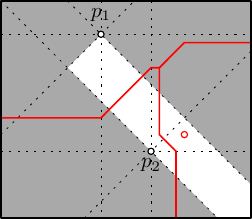}
		\caption{$p_3$ in 
			$F_9$}
		\label{fig:n=3CellI'}
	\end{subfigure}\hfil
	\begin{subfigure}{.3\textwidth}
		\centering
		\includegraphics[page=50,scale=\s]{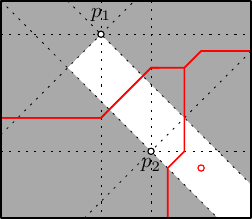}
		\caption{$p_3$ in 
			$F_{10}$}
		\label{fig:n=3CellII'}
	\end{subfigure}\hfil
	\caption{The ten combinatorially different Voronoi diagrams with no degenerate bisector \mbox{dependent} on the relative position of the third point $p_3$.}
	\label{fig:nongridn=3examples}
\end{figure}

Furthermore, at least one other point $p_2=(x_2,y_2)$ contains a corner of $R$ in its cell, without loss of generality the bottom left corner of $R$. This implies that the third point $p_3$ lies in $\bigcup_{i\in \{1,2,3,8\}}O_i(p_2)$. We distinguish the cases $x_1 > x_2$ and $x_1 \leq x_2$ which are illustrated in  \Cref{fig:nongrid3A,fig:nongrid3B} respectively. 
Moreover, the octants of $p_1$ and $p_2$ as well as the so-called partition line completing the diamond around the rightmost breakpoint of the bisector between $\B(p_1,p_2)$ (the dashed line in \Cref{fig:nongrid3A}; see Averbakh et al.~\cite{AveBerKalKra15} for more details)
subdivide the possible locations of the third point into regions which are illustrated in \Cref{fig:nongrid3}. As a result, for every position of~$p_3$ within a region, the resulting Voronoi diagram is structurally identical; see also \Cref{fig:nongridn=3examples}. Note also that all regions with the same label result in fully symmetric configurations.

For each configuration, we can describe the areas of the quarter cells dependent on the three point coordinates.
Using \ref{item:P2} and \Cref{obs:quadrant}, we obtain a number of equations.
Carrying out the involved calculations by hand is rather tedious, so we made use of MATLAB\textsuperscript{\textregistered}. It turns out that there exists a solution if and only if $\rho=\nicefrac{49}{36}$, and that $\mathcal R_3$ is the unique solution.

In the following, we present all combinatorially different Voronoi diagrams of three points containing one, two, and three degenerate bisectors, respectively. 
Exploiting \ref{item:P2} and \Cref{obs:quadrant},  we then obtain a system of equations for each diagram dependent on the points' coordinates. 
Using  MATLAB\textsuperscript{\textregistered}, we guarantee that none of the diagrams supports a balanced set.

\subparagraph{One degenerate bisector.} 
Secondly, we consider the case that $P=\{p_1,p_2,p_3\}$ contains exactly one degenerate bisector.
We may assume without loss of generality that $p_1$ has no degenerate bisector and  that
$\Delta_y(p_1,p_2)$ exceeds all of $\Delta_x(p_1,p_2)$, $\Delta_x(p_1,p_3)$, and $\Delta_y(p_1,p_3)$; otherwise we exchange the labels of $p_2$ and $p_3$ or rotate the configuration.
Furthermore, we may assume that $p_2$ lies below $p_1$ and does not lie to the right of $p_1$, see \Cref{fig:degenerateCases}.

By assumption, $p_2$ and $p_3$ share a degenerate bisector. Therefore $p_3$ lies on one of the diagonal lines through $p_2$. Moreover, $p_3$ lies above $p_2$ and not too far to the left and right of $p_1$, because of our assumption that $\Delta_y(p_1,p_2)>\Delta_x(p_1,p_3),\Delta_y(p_1,p_3)$.
As before, the octants of $p_1$ and $p_2$, as well as the partition line induced by the leftmost breakpoint of the bisector $\B(p_1,p_2)$ (represented by the dashed segment in \Cref{fig:degenerateCases} that completes the diamond about the leftmost breakpoint of $\B(p_1,p_2)$), subdivide the location of $p_3$ into seven segments which are illustrated in \Cref{fig:degenerateCases}. Placing $p_3$ on different segments results in combinatorially different Voronoi diagrams which are depicted in \Cref{fig:degnerate1}.

\begin{figure}[htbp]
	\centering
	\begin{subfigure}{\w}
		\centering
		\includegraphics[page=2,scale=\s]{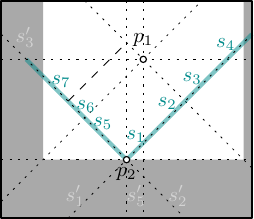}
		\caption{segments for $p_3$}
		\label{fig:degenerateCases}
	\end{subfigure}\hfil
	\begin{subfigure}{\w}
		\centering
		\includegraphics[page=3,scale=\s]{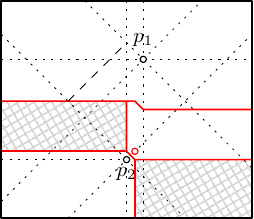}
		\caption{$p_3$ in $s_1$}
	\end{subfigure}\hfil
	\begin{subfigure}{\w}
	\centering
	\includegraphics[page=4,scale=\s]{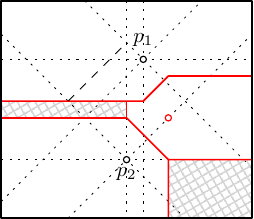}
	\caption{$p_3$ in $s_2$}
\end{subfigure}\hfil

	\begin{subfigure}{\w}
	\centering
	\includegraphics[page=5,scale=\s]{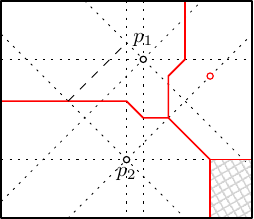}
	\caption{$p_3$ in $s_3$}
\end{subfigure}\hfil
	\begin{subfigure}{\w}
	\centering
	\includegraphics[page=6,scale=\s]{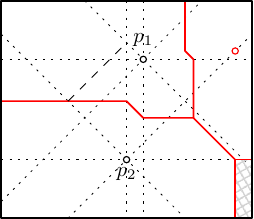}
	\caption{$p_3$ in $s_4$}
\end{subfigure}\hfil
	\begin{subfigure}{\w}
	\centering
	\includegraphics[page=7,scale=\s]{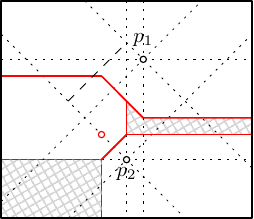}
	\caption{$p_3$ in $s_5$}
\end{subfigure}\hfil
	\begin{subfigure}{\w}
	\centering
	\includegraphics[page=8,scale=\s]{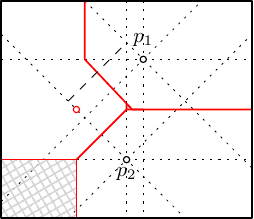}
	\caption{$p_3$ in $s_6$}
\end{subfigure}\hfil
	\begin{subfigure}{\w}
	\centering
	\includegraphics[page=9,scale=\s]{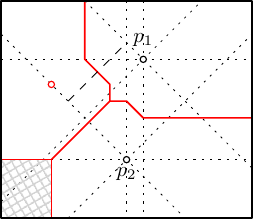}
	\caption{$p_3$ in $s_7$}
\end{subfigure}\hfil

	\caption{The seven segments producing Voronoi diagrams with one degenerate bisector.}
	\label{fig:degnerate1}
\end{figure}

\subparagraph{Two degenerate bisectors.}
Without loss of generality, we may assume that $p_3$ has a degenerate bisector with both $p_1$ and $p_2$. Since $\B(p_1,p_2)$ is non-degenerate, $p_1$ and $p_2$ lie on different diagonals and we may assume without loss of generality that $p_3$ is the rightmost point of $P$ and $\Delta_x(p_1,p_3)\geq \Delta_x(p_1,p_2)$.
Consequently, we obtain a configuration as depicted in \Cref{fig:degenerateN=3two}.

\begin{figure}[htbp]
	\centering
	\begin{subfigure}{.5\textwidth}
		\centering
		\includegraphics[page=52,scale=\s]{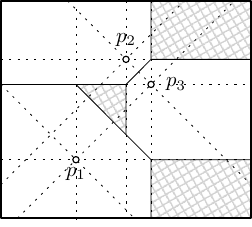}
		\caption{Two degenerate bisectors.}
		\label{fig:degenerateN=3two}
	\end{subfigure}%
	\begin{subfigure}{.5\textwidth}
		\centering
		\includegraphics[page=51,scale=\s]{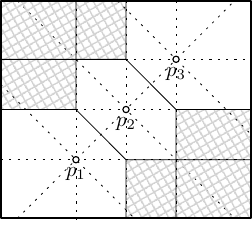}
		\caption{Three degenerate bisectors.}
		\label{fig:degenerateN=3three}
	\end{subfigure}
	\caption{Voronoi diagrams of three points containing two degenerate bisectors.}
	\label{fig:degenerateTwoThreee}
\end{figure}

\subparagraph{Three degenerate bisectors.}
If all three bisectors are degenerate, then the three points lie on a common diagonal line: consider two points on a line of slope~$+1$. Since the diagonals through the two points of slope~$-1$ are parallel, the third point must lie on the diagonal of slope~$+1$; see \Cref{fig:degenerateN=3three}.  Without loss of generality (allowing for reflection and rotation and renaming the points) we may assume the illustrated labelling.

In showing that the resulting systems of equations have no solutions, we prove that  $\mathcal R_3$ is the unique balanced non-grid configuration with three points.
\end{claimproof}

\subsection{An Infinite Family of Balanced Configurations}\label{sec:balanced_family}
Observe that $\mathcal R_{2,\rho},\mathcal R_{3},\mathcal R_{4}$, and $\mathcal R_{5}$ are \emph{atomic}, i.e., they cannot be
decomposed into subconfigurations whose union of Voronoi cells is a rectangle. 
We show how they serve as building blocks to induce large families of balanced configurations.

\begin{theorem}\label{lem:existence}
	For every integer $n\geq 2$,  there exists a rectangle $R$ and a set $P$ of $n$ points 
such that $P$ is \balanced and no Voronoi cell is a rectangle. 
\end{theorem}
\begin{proof}
For every $n=3k+5\ell$ with $k,\ell\in\{0,1,\ldots\}$, we construct a 
configuration by  combining $k$ blocks of $\mathcal R_3$ and $\ell$ blocks of $\mathcal R_5$, as illustrated in \Cref{fig:equalAreaN}. This yields configurations with $n$ points in which
$n=3k$ for $k\geq 1$, $n=3k+2=3(k-1)+5$ for $k\geq 1$, or 
$n=3k+1=3(k-3)+10$  for $k\geq 3$, so we obtain configurations for all $n\geq 8$ and $n=3,5,6$.

\begin{figure}[htbp]
	\centering\begin{subfigure}[t]{.9\textwidth}
		\centering
		\includegraphics[page=23]{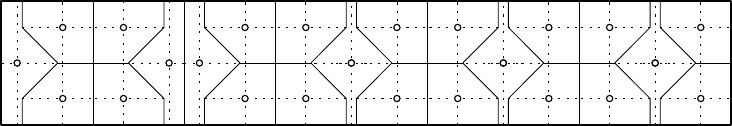}
		\caption{Combining $k$ blocks of $\mathcal R_3$ and $l$ blocks of $\mathcal R_5$ for a configuration with $n=3k+5l$ points in a rectangle $R$ with $\rho(R)=\nicefrac{1}{49}(36k+60l)$.}
		\label{fig:equalAreaN}
	\end{subfigure}
    \vspace{6pt}
    
	\begin{subfigure}[t]{.45\textwidth}
		\centering
		\includegraphics[page=24]{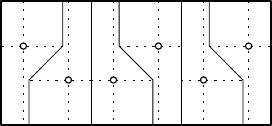}
		\caption{Combining $k$ blocks of $\mathcal R_2$ for a configuration with $n=2k$ points.}
		\label{fig:equalArea2N}
	\end{subfigure}\hfill
		\begin{subfigure}[t]{.45\textwidth}
		\centering
		\includegraphics[page=65]{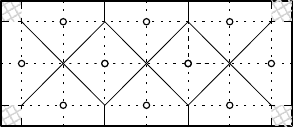}
		\caption{Combining $k$ blocks of $\mathcal R_4$ with partial overlap for a configuration with $n=3k+1$ points.
		}
		\label{fig:equalArea7}
	\end{subfigure}
	\caption{Illustration of the proof of \cref{lem:existence}.}
	\label{key}
\end{figure}

Balanced configurations with $n=2k$, $k\in\mathbb N$, points are obtained by combining $k$ blocks of~$\mathcal R_2$ as shown 
in \Cref{fig:equalArea2N}; alternatively, for the missing cases of $n=2,4$, recall the configurations in \Cref{fig:equalArea}.

Lastly, \cref{fig:equalArea7} depicts a balanced configuration for the case of $n=3k+1$ points by combining $k$ blocks of $\mathcal R_4$ with partial overlap. In particular, this contains the last missing case of  
 $n=7$ points. Note that, in contrast to the previous configurations, these configurations contain degenerate bisectors and neutral regions. (We do not know of examples for $n=7$ without degenerate bisectors.)
\end{proof}

While none of the configurations in \Cref{lem:existence} contains a rectangular Voronoi cell,
they contain many immediate repetitions of the same atomic components. In fact, there are arbitrarily large \emph{non-repetitive} balanced configurations without
directly adjacent congruent atomic subconfigurations. 

\begin{theorem}
\label{th:irregular}
There is an injection between the family of 0-1 strings
and a family of non-repetitive \balanced configurations without any rectangular Voronoi cells.
\end{theorem}

\begin{proof}
For a given 0--1 string $\mathcal S$ of length $s$, we use $s$ pairs of blocks $\mathcal R_3$ and its
reflected version $\mathcal R_3'$ to build a sequence of $2s$ blocks. 
As shown in \Cref{fig:blocks678}, 
we insert a block~$\mathcal R_5$ after the $i$th pair if $\mathcal S$
has a $1$ in position $i$; otherwise the block sequence remains.
\begin{figure}[htb]
	\centering
	\centering
	
	\includegraphics[page=35]{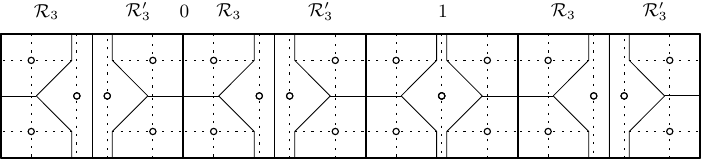}
	\caption{Illustration of the proof of \Cref{th:irregular}.
		 The configuration represents the string 01.
	}
	\label{fig:blocks678}
\end{figure}
\end{proof}

%% file: 03-Black.tex
\section{The One-Round Manhattan Voronoi Game}\label{sec:black}

An instance of the One-Round Manhattan Voronoi Game consists of a rectangle~$R$ and the number $n$ of points to be played by each player. 
Without loss of generality, $R$ has height~$1$ and width $\rho\geq 1$.
White chooses a set $W$ of $n$ white
points  in~$R$, followed by Black selecting
a set $B$ of $n$ black points, with $W\cap B=\varnothing$. Each player scores the area consisting of the points that are closer to one of their facilities than to any one of their opponent's.
Hence, if two points of one player share a degenerate bisector, the possible neutral regions are assigned to this player. Therefore, by replacing each degenerate bisector between points of one player with a (w.l.o.g. horizontal) non-degenerate bisector, each player scores the area of its (horizontally enlarged) Manhattan Voronoi cells.
With slight abuse of notation, we denote the resulting (horizontally enlarged) Voronoi cells of coloured point sets by $V^{W\cup B}(p)$ in the same way as before.
The player with the higher score wins, or
the game ends in a tie.

For an instance $(R,n)$ and a set $W$ of $n$ 
white points, a set $B$ of $n$ black points is a \emph{winning set for
Black} if Black wins the game by playing $B$; 
likewise, $B$ is a \emph{tie set} if the game ends in a tie. \new{For a given set $W$ of white points,}  a black point~$b$
is a \emph{winning point} if 
\new{$\area(V^{W\cup \{b\}}(b))$} exceeds $\nicefrac{1}{2n}\cdot \area(R)$. 
A white point set $W$ is \emph{unbeatable} if it does not admit a winning set for Black, and $W$ is  a  \emph{winning set \new{for White}} if there exists neither a tie nor a winning set for Black.
If Black or White can always identify a winning set, 
we say that they have a \emph{winning strategy}.

Despite the possible existence of degenerate bisectors for Manhattan distances,  we show that Black has a winning strategy if and only if
Black has a winning point. 
We make use of the following two lemmas.

\begin{restatable}{lemma}{steal}\label{lem:steal}
	Consider a rectangle $R$ with a set $W$ of white points. Then for every $\varepsilon>0$ and every half cell $H$ of $W$, Black can place a point $b$ such that the area of
	$V^{W\cup\{b\}}(b)\cap H$
	 is at least $\left(\area(H)-\varepsilon\right)$.
\end{restatable}
\begin{proof}
	Without loss of generality, we consider the left half cell $H$ of some $w\in W$ as in \Cref{fig:stealing}. 
	\begin{figure}[b]
		\centering
		\begin{subfigure}[t]{.45\textwidth}
			\centering
			\includegraphics[page=1]{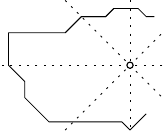}
			\caption{A left half cell $H$.}
			\label{fig:cell}
		\end{subfigure}
		\hfil
		\begin{subfigure}[t]{.45\textwidth}
			\centering
			\includegraphics[page=2]{Cell}
			\caption{A point capturing $H$ up to any $\varepsilon>0$.
			}
			\label{fig:cellBlack}
		\end{subfigure}
		\caption{Illustration of the proof of \Cref{lem:steal}.}
		\label{fig:stealing}
	\end{figure}
	By placing $b$ slightly to the left of $w$, the bisector~$\B(b,w)$ is a vertical segment between $b$ and $w$. Therefore $V^{W\cup\{b\}}(b)$ contains all points of $H$ to the left of $\B(b,w)$. The area difference between $V^{W\cup\{b\}}(b)\cap H$ and $H$ is bounded above by the product of the total length of the top and bottom arm of~$H$ and (half) the distance of $b$ and~$w$. Consequently, by placing $b$ close enough, the difference drops below any fixed $\varepsilon>0$.
\end{proof}

In fact, White must play a balanced set; otherwise Black can win.
\begin{restatable}{lemma}{halving}\label{lem:halving}
	Let $W$ be a set of $n$ white points in a rectangle $R$. If any half cell of~$W$ has an area different from $\nicefrac{1}{2n}\cdot\area(R)$, then Black has a winning strategy. 
\end{restatable}
\begin{proof}
	If not all half cells of $W$ have the same area, then
	there exists a half cell~$H$ with $\area(H)>\nicefrac{1}{2n}\cdot \area(R)$. 
	We assume w.l.o.g.~that $H$ is a half cell of $\mathcal H^\vert$; otherwise we consider $\mathcal H^-$.
	Denoting the $n$ largest half cells of $\mathcal H^\vert$ by $H_1,\dots,H_n$, it follows that there exists $\delta>0$ such that
	$\sum_{i=1}^n H_i= \nicefrac{1}{2}\cdot\area(R) +\delta.$
	
	By \Cref{lem:steal}, Black can place a point $b_i$ to capture the area of $H_i$ up to any $\varepsilon>0$. More precisely, by choosing $\varepsilon < \nicefrac{\delta}{n}$, \Cref{lem:steal} guarantees that there exists a placement of $n$ black points $b_1,\dots,b_n$ such that 
	\begin{align*}
	\sum_{i=1}^n\area(V^{W\cup B}(b_i)\cap H_i) 
	&\geq \sum_{i=1}^n{(\area(H_i)-\varepsilon)}
	= \frac{1}{2} \area(R)+\delta-n\varepsilon
	>  \frac{1}{2} \area(R).
	\end{align*}
	Consequently, Black has a winning strategy by placing these points.
\end{proof}

These insights enable us to prove the main result of this section.

\begin{restatable}{theorem}{black}\label{thm:black}
	Black has a winning strategy for a set $W$ of $n$ white points in a rectangle $R$ if and only if Black has a winning point.
\end{restatable}
\begin{proof}
\new{Let $B$ be a winning set for Black.}
If \new{Black's winning score with $B$} exceeds
$\nicefrac{1}{2}\cdot \area(R)$ 
  then, by the pigeonhole principle, \new{there exists  $b\in B$ such that its cell $V^{W\cup B}(b)$ has an area exceeding}
  $\nicefrac{1}{2n}\cdot \area(R)$. \new{Since $\area(V^{W\cup\{b\}}(b))\geq \area(V^{W\cup B}(b))$}, $b$ is  a winning point.
  Otherwise Black \new{wins by playing $B$ but} scores at most half the area of $R$. 
\new{We show that there exists a winning set for Black which achieves a higher score, enabling us to use the argument presented above given any winning set.}

\begin{claim}
	Let $B$ be a winning set for Black such that Black scores at most $\nicefrac{1}{2n}\cdot \area(R)$. Then there is also a winning set $B'$ for Black such 
	that Black's score exceeds $\nicefrac{1}{2n}\cdot \area(R)$.
\end{claim}
\begin{claimproof}
	If Black wins with $B$, then Black's score exceeds White's score. Moreover, since Black scores at most $\nicefrac{1}{2n}\cdot \area(R)$, there exist neutral zones and degenerate bisectors between black and white points.
	
	Consider a white point $w$ and a black point $b$ with $ \Delta_x(w,b)= \Delta_y(w,b)$.
	Black can avoid this degeneracy by choosing a slightly perturbed location. By moving on either side of the diagonal line through $w$ and $b$, Black can win either of the neutral \new{regions' areas} up to any $\varepsilon>0$. If the neutral regions are of different sizes, then Black can ensure a net gain. If the areas are 
	the same, then Black has a net loss of $\varepsilon>0$. However, since Black wins, they can allow for some net loss $\varepsilon>0$. 
	This argument applies even if $b$ contributes to more than one degeneracy by consideration of the sum of the losses and gains in the resulting cells. Therefore $b$ can avoid degeneracy by an arbitrarily small net loss. The repeated application of these perturbations for all of Black's points shows that Black has a winning set without forcing neutral regions. Consequently, Black's score exceeds $\nicefrac{1}{2n}\cdot \area(R)$.
\end{claimproof}

Now suppose that there exists a winning point $b$, i.e.,  $\area(V^{W\cup\{b\}}(b))=\nicefrac{1}{2n}\cdot\area(R)+\delta$ for some $\delta>0$.  If $n=1$, Black clearly wins with $b$. If $n\geq 2$,  Black places $n-1$ further black points:
consider $w_i\in W$. By \Cref{lem:halving}, we may assume that each half cell of $W$ has area $\nicefrac{1}{2n}\cdot \area(R)$.   By \Cref{obs:prop}(\ref{itemA}), $w_i$ has a half cell $H_i$ that is disjoint from $V^{W\cup\{b\}}(b)$.
	 By \Cref{lem:steal}, Black can place a point $b_i$ 
	 to capture the area of $H_i$ up to every $\varepsilon>0$. Choosing $\varepsilon<\nicefrac{\delta}{n-1}$ and placing one black point for \new{every} $n-1$ distinct white point\new{,} with \Cref{lem:steal}, 
	 Black achieves a score of
	 $\sum_{p\in B}\area(V^{W\cup B}(p)) =\left(\nicefrac{1}{2n}\cdot \area(R)+\delta\right)+(n-1)\left(\nicefrac{1}{2n}\cdot \area(R)-\varepsilon\right)
	 >\nicefrac{1}{2}\cdot \area(R).$
	Consequently, Black has a winning strategy. 
\end{proof}

%% file: 04-White.tex
\section{Properties of Unbeatable Sets}\label{sec:white}

In this section, we identify necessary properties of \emph{unbeatable} white sets, for which the game ends in a tie or White wins. 
We call a cell a \emph{\bridge} if it has two opposite boundary arms.

\begin{theorem}\label{thm:properties}
	If $W$ is an unbeatable white point set in a rectangle $R$, then it fulfils the following properties:
	\begin{description}
	\item[(P1)\label{item:P1}] The area of every half cell of $W$ is $\nicefrac{1}{2n}\cdot\area(R)$.
	\item[(P2)\label{item:P2}] The arms of a non-\bridge cell are equally long; the opposite boundary arms of a \bridge cell are of equal length and, if $|W|>1$, they are shortest among all arms.
\end{description}
\end{theorem}

\begin{proof}
	Because $W$ is unbeatable, property \ref{item:P1} follows immediately from \Cref{lem:halving}. Moreover, in case $|W|=1$, \ref{item:P1} implies that opposite arms of the unique (bridge) cell have equal length, i.e., \ref{item:P2} holds for $|W|=1$.
	
	It remains to prove property \ref{item:P2} for $|W|\geq 2$. By \Cref{thm:black}, it suffices to identify a black winning point if \ref{item:P2} is violated. We start with the following fact.

\begin{restatable}{claim}{shift}\label{clm:DeltaShift}
	Let $P$ be a point set containing $p=(0,0)$ and let $P'$ be obtained from $P$ by adding $p'=(\delta,\delta)$ where $\delta>0$ such that $p'$ lies within $V^P(p)$.
	Restricted to $Q:=Q_1(p')$, the cell $V^{P'}(p')$ contains all points that are obtained when the boundary of $V^{P}(p)\cap Q$ is shifted upwards (rightwards) by $\delta$ (if it does not intersect the boundary of $R$). 
\end{restatable}

\begin{claimproof}
	To prove this claim, it suffices to consider the individual bisectors of $p'$ and any other point $q\in P'$. Note that all points shaping the cell of $p'$ in quadrant $Q$ are contained in an octant $O_i(p')$ with $i\in\{1,2,3,8\}$.
	We show that vertical bisectors move rightwards and horizontal bisectors move upwards.
	
	For a point $q=(x,y)$ in $\overline{O}_2(p')$, the part of the bisector $\B(q,p')\cap Q$ can be obtained from $\B(q,p)\cap Q$ by shifting it upwards by an amount of $\delta$; see also \Cref{fig:O2}. In particular, the initial height of the diagonal segment remains unchanged because its vertical distance to $q$ is $\nicefrac{1}{2}(\Delta_y(q,p')-\Delta_x(q,p'))=\nicefrac{1}{2}((y-\delta)-(x-\delta))=\nicefrac{1}{2}(y-x)=\nicefrac{1}{2}(\Delta_y(q,p)-\Delta_x(q,p))$. Note also that this holds for degenerate bisectors, because only their diagonal segment is contained in $Q$.
	
	\begin{figure}[htb]
		\centering
		\begin{subfigure}[t]{.25\textwidth}
			\centering
			\includegraphics[page=8]{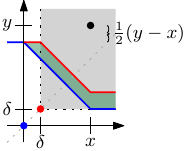}
			\caption{$q\in \overline{O}_2(p')$}
			\label{fig:O2}
		\end{subfigure}\hfil
		\begin{subfigure}[t]{.25\textwidth}
			\centering
			\includegraphics[page=10]{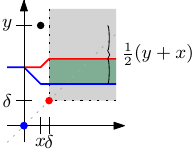}
			\caption{$q\in \overline{O}_3(p')\cap O_2(p)$}
			\label{fig:O3a}
		\end{subfigure}\hfil
		\begin{subfigure}[t]{.3\textwidth}
			\centering
			\includegraphics[page=9]{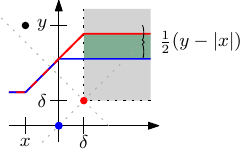}
			\caption{$q\in \overline{O}_3(p')\cap \overline{O}_3(p)$}
			\label{fig:O3b}
		\end{subfigure}
		\caption{Illustration of \cref{clm:DeltaShift}. If $q=(x,y)$ lies in 
			$\overline O_2(p')\cup\overline O_3(p')$, the part of the bisector~$\B(q,p')$ within the first quadrant $Q$ of $p'$ coincides with $\B(q,p)\cap Q$ shifted upwards by~$\delta$.}
		\label{fig:shifts}
	\end{figure}

	For a point $q=(x,y)$ in $\overline{O}_3(p')\cap O_2(p)$, the vertical distance of $q$ to the horizontal segment of $\B(q,p)$ within $Q$ is $\nicefrac{1}{2}(y-x)+x=\nicefrac{1}{2}(y+x)$ while vertical distance of $q$ to the horizontal segment of $\B(q,p')$ is $\nicefrac{1}{2}((y-\delta)-(\delta-x))=\nicefrac{1}{2}(y+x)-\delta$; see also \Cref{fig:O3a}. 
	
	For a point $q=(x,y)$ in $\overline{O}_3(p')\cap \overline{O}_3(p)$, the vertical distance of $q$ to the horizontal segment of the bisector within $Q$ is $\nicefrac{1}{2}((y-\delta)-(|x|+\delta))=\nicefrac{1}{2}(y-|x|)-\delta$ for $p'$ and $\nicefrac{1}{2}(y-|x|)$ for $p$; see also \Cref{fig:O3b}. 
	
	Note that for $q\in{O}_2(p')\cup {O}_3(p')$, the bisector $\B(q,p)$ is horizontal. Consequently, all shifted segments are horizontal or diagonal. Shifting them rightwards yields a region contained in $V^{P'}(p')$.  By symmetry, all (vertical) bisectors of points within ${O}_1(p')\cup {O}_8(p')$ are shifted rightwards. This implies the claim.
\end{claimproof}

We use our insight of \Cref{clm:DeltaShift} to show property  \ref{item:P2} in two steps.

	\begin{claim}\label{clm:armsAll}
	Let $w\in W$ be a point such that an arm $A_1$ of $V^W(w)$ is shorter than a neighbouring arm~$A_2$ and the arm~$A_3$ opposite to $A_1$ is inner. Then Black has a winning point.
	\end{claim}
\begin{claimproof}
Without loss of generality, we consider the case that $A_1$ is the bottom arm of $V^W(w)$,  $A_2$ its right arm, and $w=(0,0)$; see \Cref{fig:TopLeftArm}. We denote the length of $A_i$ by $|A_i|$. 
Now we consider Black placing a point $b$ within $V^W(w)$ at $(\delta,\delta)$ for some $\delta>0$. 
To ensure that the cell of $b$ contains almost all of the right half cell of $V^W(w)$, we infinitesimally perturb $b$ rightwards; for ease of notation in the following analysis, we omit the corresponding infinitesimal terms and assume that the bisector of $b$ and $w$ is vertical.  We compare the area of $V(b):=V^{W\new{\cup\{b\}}}(b)$  with the right half cell $H$ of $w$. In particular, we show that there exists $\delta>0$ such that the area of $V(b)$ exceeds the area of~$H$. Because $\area(H)=\nicefrac{1}{2n}\cdot \area(R)$ by~\ref{item:P1}, $b$ is a winning point.

	\begin{figure}[htb]
		\centering
		\begin{subfigure}[t]{.243\textwidth}
			\centering
			\includegraphics[page=3]{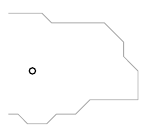}
			\caption{The right half cell $H$.}
			\label{fig:LeftCell}
		\end{subfigure}\hfill
		\begin{subfigure}[t]{.35\textwidth}
			\centering
			\includegraphics[page=4]{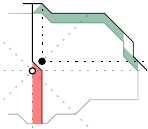}
			\caption{Case: The bottom arm $A_1$ is shorter than the right arm $A_2$ and the top arm $A_3$ is inner.}
			\label{fig:TopLeftArm}
		\end{subfigure}\hfill
		\begin{subfigure}[t]{.3\textwidth}
			\centering
			\includegraphics[page=6]{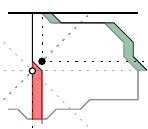}
			\caption{Case: The bottom arm $A_1$ is shorter than the top arm $A_3$ and the right arm $A_2$ is inner.}
			\label{fig:TopBottomArm}
		\end{subfigure}
		\caption{Illustration of \cref{clm:armsAll} and \cref{clm:armsBridge}: the gain and loss of $V(b)$ compared to $H$.}
		\label{fig:armClaims}
	\end{figure}

Clearly, all points in $H$ to the right of the (vertical) bisector of $b$ and $w$ are closer to $b$.
Consequently, when compared to $H$, the loss of $V(b)$ is upper bounded by $\delta|A_1|+\nicefrac{1}{2} \delta^2$; see also \Cref{fig:TopLeftArm}.
By  \Cref{clm:DeltaShift} and the fact that $A_3$ is inner, $V(b)\cap Q_1(p)$ gains at least $\delta(|A_2|-\delta)$ when compared to $H\cap Q_1(p)$. 
When additionally guaranteeing $\delta<\nicefrac{2}{3}(|A_2|-|A_1|)$, the gain exceeds the loss and thus $b$ is a winning point.
\end{claimproof}

For a cell with two neighbouring inner arms,  
\Cref{clm:armsAll} implies that all its arms have equal length. Consequently, it only remains to prove \ref{item:P2} for \bridge{s}.
With arguments similar to those proving \Cref{clm:armsAll}, we obtain the following result. For an illustration, see \Cref{fig:TopBottomArm}. 

\begin{restatable}{claim}{armsBridge}\label{clm:armsBridge}
	If there exists a point $w\in W$ such that two opposite arms of $V^W(w)$  have different lengths and a third arm is inner, then Black has a winning point.
\end{restatable}
\begin{claimproof}
	Without loss of generality, we consider the case that $A_1$ is the bottom arm, $A_1$ is shorter than the top arm $A_3$, and the right arm $A_2$ is inner. 
	Analogously to the proof of \Cref{clm:armsAll}, Black places a point $b$ at $(\delta,\delta)$ for some $\delta>0$ and chooses the vertical bisector with $w$. As above, when compared to the right half cell $H$ of $w$, the loss of $V^{W\new{\cup\{b\}}}(b)$  is bounded above by $\delta|A_1|+\nicefrac{1}{2} \, \delta^2$.
	By \Cref{clm:DeltaShift} and the fact that $A_2$ is inner, the gain is bounded below by $\delta(|A_3|-\delta)$.
	Guaranteeing $\delta<\nicefrac{2}{3}(|A_3|-|A_1|)$, the gain exceeds the loss. Thus, if $|A_3|>|A_1|$, Black has  a winning point.
\end{claimproof}

If $|W|>1$, every cell has at least one inner arm. Therefore \Cref{clm:armsBridge} yields that opposite boundary arms of a bridge cell have equal length. Moreover, \Cref{clm:armsAll} implies that the remaining arms are not shorter. This proves \ref{item:P2} for bridges.
\end{proof}

We now show that unbeatable white sets are grids; in some cases they are even \emph{square grids}, i.e., every cell is a square.

\begin{lemma}\label{lem:grids}
Let $P$ be a set of $n$ points in a $(1\times\rho)$ rectangle $R$ with $\rho\geq 1$ fulfilling properties \ref{item:P1} and \ref{item:P2}. Then $P$ is a grid. More precisely, if $\rho\geq n$, then $P$ is a $1\times n$ grid; otherwise, $P$ is a square grid.
\end{lemma}
\begin{proof}
We distinguish two cases.

Case 1: $\rho\geq n$.
By \ref{item:P1},  every half cell has area 
$\nicefrac{1}{2n}\cdot \area(R)=\nicefrac{1}{2n}\cdot \rho \geq \nicefrac{1}{2}$. Since the height of every half cell is bounded by $1$, every left and right arm has a length of at least~$\nicefrac{1}{2}$. 
Then, property \ref{item:P2} implies that each top and bottom arm has length $\nicefrac{1}{2}$, i.e., every $p\in P$ is placed on the horizontal centre line of $R$.
Finally, again by \ref{item:P1}, the points must be evenly spread.
Hence,   $P$ is a $1\times n$ grid.

Case 2: $\rho< n$. 
We consider the point $p$ whose cell $V^P(p)$ contains the top left corner of~$R$ and denote its quarter cells by $C_i$. Then, $C_2$ is a rectangle.
Moreover,  $V^P(p)$ is not a \bridge; otherwise  
its left half cell has area \new{at least} $\nicefrac{1}{2}> \nicefrac{1}{2n}\cdot \rho=\nicefrac{1}{2n}\cdot \area(R)$. Therefore, by  \ref{item:P2}, all arms of $V^P(p)$ have the same length; we denote this length by $d$. Together with the fact that $C_2$ and $C_4$ have the same area by \ref{item:P1} and \Cref{obs:quadrant}, it follows that $C_2$ and $C_4$  are squares of side length~$d$.

We consider the right boundary of $C_4$. 
Since the right arm of $V^P(p)$ has length~$d$ (and the boundary continues vertically below), some point $q$ has distance~$2d$ to $p$ and lies in $Q_1(p)$. The set of all these  possible point locations forms a segment, which is highlighted in red in \Cref{fig:gridA}. 
Consequently, the left arm of $q$ has length  $d$. 
By \ref{item:P2}, the top arm of $q$ must also have length $d$. Hence, $q$ lies at the grid location illustrated in \Cref{fig:gridB}. Moreover, it follows that $q$ is the unique point whose cell shares part of the boundary with $C_1$; otherwise the top arm of $q$ does not have length $d$.

\begin{figure}[htbp]
	\centering
	\begin{subfigure}[t]{.25\textwidth}
		\centering
		\includegraphics[page=10]{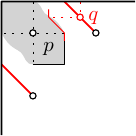}
		\caption{Points $q \in Q_1(p) \cup Q_3(p)$ with $l_1(p,q)=2d$.}
		\label{fig:gridA}
	\end{subfigure}\hfil
	\begin{subfigure}[t]{.25\textwidth}
		\centering
		\includegraphics[page=11]{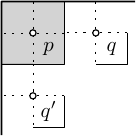}
		\caption{Unique neighbouring cell generators.}
		\label{fig:gridB}
	\end{subfigure}\hfil
	\begin{subfigure}[t]{.3\textwidth}
		\centering
		\includegraphics[page=9]{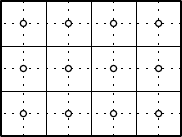}
		\caption{Final square grid.}
		\label{fig:gridC}
	\end{subfigure}
	\caption{Illustration of the proof of \Cref{lem:grids}.}
	\label{fig:grid2}
\end{figure}

By symmetry, a point $q'$ lies at a distance $2d$ below $p$ and distance $d$ to the boundary. 
Thus, every quarter cell of $V^P(p)$ is a square with edge length $d$; hence, the arms of all cells have length at least $d$.  Moreover, the top left quarter
cells of $V^P(q)$ and $V^P(q')$ are squares, so their bottom right
quadrants must also be squares. Using this argument iteratively along the
boundary implies that boundary cells are squares. Applying it to the
remaining rectangular hole shows that $P$ is a square grid.
\end{proof}

We now come to our main result. 
\begin{theorem}\label{thm:winningStrategy}
White has a winning strategy for placing $n$ points in a $(1\times \rho)$ rectangle with $\rho\geq 1$ if and only if $\rho\geq n$; otherwise Black has a winning strategy.
Moreover, if $\rho\geq n$, the unique winning strategy for White is to place a $1\times n$ grid.
\end{theorem}
\begin{proof}
First we show that Black has a winning strategy if $\rho< n$. Suppose that Black cannot win. Note that $\rho< n$ implies $n\geq 2$.  Consequently, by \Cref{thm:properties} and \Cref{lem:grids}, the white point set $W$ is a square $a\times b$ grid with $a,b\geq 2$, and thus the four cells in the top left corner induce a $2\times 2$ grid. By \Cref{thm:black}, it suffices to identify a winning point for Black. Thus, we show the following:
\begin{claim}
	Black has a winning point in a square $2\times 2$ grid.
\end{claim}
\begin{claimproof}
	Suppose the arms of all cells have length~$d$. Then a black point $p$ is a winning point if its cell has an area exceeding $2d^2$. With $p$ placed at a distance $\nicefrac{3d}{2}$ from the top and left boundary as depicted in \Cref{fig:2by2}, the cell of $p$ has an area of $2d^2+\nicefrac{d^2}{4}$.
\end{claimproof}

\begin{figure}[htbp]
	\centering
	\begin{subfigure}[t]{.2\textwidth}
		\centering
		\includegraphics[page=12,scale=.97]{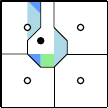}
		\caption{A black winning point in a $2\times 2$ grid.}
		\label{fig:2by2}
	\end{subfigure}
	\hfill
	\begin{subfigure}[t]{.78\textwidth}
		\centering
		\includegraphics[page=7,scale=.97]{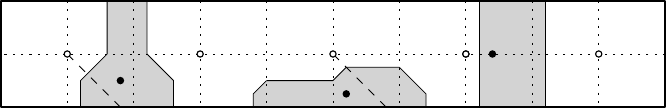}
		\caption{Every black cell has an area $\leq \nicefrac{1}{2n}\cdot \area(R)$. Moreover,  only $n-1$ locations result in cells of that size. }
		\label{fig:gridBlack}
	\end{subfigure}
	\caption{Illustration of the proof of \cref{thm:winningStrategy}.
	}
\end{figure}

Secondly, we consider the case $\rho\geq n$ and  show that White has a winning strategy.  \Cref{thm:properties} and \Cref{lem:grids} imply that White must place its points in a $1\times n$ grid; otherwise Black can win. We show that Black has no option to beat this placement; i.e., if $\rho\geq n$, then:
\begin{claim}
	Black has no winning point and cannot force a tie in a $1\times n$ grid.
\end{claim}
\begin{claimproof}
	By symmetry, there essentially exist two different placements of a black point~$b$ with respect to a closest white point $w_b$. Without loss of generality, we assume that $w_b$ is to the left and not below $b$. Let $x$ and $y$ denote the horizontal and vertical distance of $b$ to $w_b$, respectively. For a unified presentation, we add half of potential neutral zones in case $x=y$ to the area of the black cell. As a consequence, Black loses if its cells have an area of less than $\nicefrac{1}{2}\cdot\area(R)$.
	
	If $x> y$, the cell of $b$ evaluates to an area of (at most) $\nicefrac{1}{2n}\cdot\area(R)-y^2$. In particular, it is maximized for $y=0$, i.e., when $b$ is placed on the horizontal centre line of $R$ and if there exist white points to the left and right of $b$. In this case the cell area is exactly $\nicefrac{1}{2n}\cdot\area(R)$.
	
	If $x\leq y$, the cell area of $b$ has an area of (at most)
	$\nicefrac{1}{2n}\cdot\area(R)-y(w'-h')-\nicefrac{1}{4}(3y^2+x^2)$,
	where $w':=\nicefrac{w}{2n}$ and $h':=\nicefrac{h}{2}$ denote the dimensions of the grid cells. Note that $w'\geq h'$ because $\rho\geq n$. 
	Consequently, the cell area is maximized for $x=0,y=0$. However, this placement coincides with the location of a white point and is thus forbidden. Therefore every valid placement results in a cell area strictly smaller than $\nicefrac{1}{2n}\cdot\area(R)$. 
	Consequently, Black has no winning point.
	
	Note that the cell area is indeed strictly smaller than the abovementioned maximum values if the black point does not have white points on both sides. Therefore the (unique) best placement of a black point is on the centre line between two white points, as illustrated by the rightmost black point in \Cref{fig:gridBlack}. However, there exist only $n-1$ distinct positions of this type; all other placements result in strictly smaller cells. Consequently, Black cannot force a tie and so loses.
	\end{claimproof}
	This completes the proof of the theorem.
\end{proof}

%% file: 06-Conclusion.tex
\section{Open Problems}\label{sec:conclusion}

There are various directions for future work.

We demonstrated that there is a spectrum of balanced configurations,
based on identifying a number of small \emph{atomic} (i.e., non-decomposable) 
configurations that can be concatenated in a strip-like fashion. 
Are there further atomic configurations? Is it possible to combine them
into more intricate two-dimensional patterns rather than just putting together
identical strip-based configurations? Beyond that, the biggest challenge
is clearly to provide a full characterization of balanced configurations,
with further generalizations to other metrics and dimensions.

As our main result, we presented a full characterization of the {One-Round Voronoi Game} with Manhattan distances.
Just as for the previously studied Euclidean metric, this still leaves the multi-round variant as a wide open (and, most probably, quite difficult) problem. Further interesting problems arise from considering higher-dimensional variants.